\documentclass[11pt]{article}

\usepackage{etoolbox}

\usepackage[utf8]{inputenc} %
\usepackage[T1]{fontenc}    %
\usepackage{url}            %
\usepackage{booktabs}       %
\usepackage{amsfonts}       %
\usepackage{nicefrac}       %
\usepackage{microtype}      %
\usepackage{xspace}
\usepackage{complexity}
\usepackage{enumitem}

\usepackage{natbib}

\usepackage{algpseudocode,algorithm,algorithmicx}
\usepackage{verbatim}
\usepackage{inconsolata}

\usepackage{setspace}
\usepackage{subfigure}
\usepackage{amsmath,bm}

\usepackage{optidef}

\usepackage[letterpaper, left=1in, right=1in, top=1in, bottom=1in]{geometry}

\usepackage[dvipsnames]{xcolor}
\usepackage[colorlinks=true, linkcolor=blue!70!black, citecolor=blue!70!black,urlcolor=black,breaklinks=true]{hyperref}
\usepackage{microtype}

\usepackage{algorithm}

 \usepackage{natbib}
 \bibpunct{(}{)}{;}{a}{,}{,}

\usepackage{amsthm}
\usepackage{mathtools}
\usepackage{amsmath}
\usepackage{bbm}
\usepackage{amsfonts}
\usepackage{amssymb}
\let\vec\undefined

\usepackage{xpatch}

\theoremstyle{definition}  %

\newtheorem{lemma}{Lemma}[section]
\newtheorem{conjecture}{Conjecture}
\newtheorem{corollary}{Corollary}[section]
\newtheorem{proposition}{Proposition}[section]
\newtheorem{fact}{Fact}
\newtheorem{assumption}{Assumption}

\theoremstyle{plain}
\newtheorem{question}{Question}
\newtheorem{remark}{Remark}

\newtheorem{theorem}{Theorem}[section]
\newtheorem{definition}{Definition}[section]

\xpatchcmd{\proof}{\itshape}{\normalfont\proofnameformat}{}{}
\newcommand{\proofnameformat}{\bfseries}

\makeatletter
\newcommand{\neutralize}[1]{\expandafter\let\csname c@#1\endcsname\count@}
\makeatother

\usepackage{prettyref}

\newcommand{\savehyperref}[2]{\texorpdfstring{\hyperref[#1]{#2}}{#2}}
\newrefformat{eq}{\savehyperref{#1}{\textup{(\ref*{#1})}}}
\newrefformat{eqn}{\savehyperref{#1}{Equation~\ref*{#1}}}
\newrefformat{lem}{\savehyperref{#1}{Lemma~\ref*{#1}}}
\newrefformat{def}{\savehyperref{#1}{Definition~\ref*{#1}}}
\newrefformat{line}{\savehyperref{#1}{line~\ref*{#1}}}
\newrefformat{thm}{\savehyperref{#1}{Theorem~\ref*{#1}}}
\newrefformat{corr}{\savehyperref{#1}{Corollary~\ref*{#1}}}
\newrefformat{cor}{\savehyperref{#1}{Corollary~\ref*{#1}}}
\newrefformat{sec}{\savehyperref{#1}{Section~\ref*{#1}}}
\newrefformat{app}{\savehyperref{#1}{Appendix~\ref*{#1}}}
\newrefformat{ass}{\savehyperref{#1}{Assumption~\ref*{#1}}}
\newrefformat{asm}{\savehyperref{#1}{Assumption~\ref*{#1}}}
\newrefformat{ex}{\savehyperref{#1}{Example~\ref*{#1}}}
\newrefformat{fig}{\savehyperref{#1}{Figure~\ref*{#1}}}
\newrefformat{alg}{\savehyperref{#1}{Algorithm~\ref*{#1}}}
\newrefformat{rem}{\savehyperref{#1}{Remark~\ref*{#1}}}
\newrefformat{conj}{\savehyperref{#1}{Conjecture~\ref*{#1}}}
\newrefformat{prop}{\savehyperref{#1}{Proposition~\ref*{#1}}}
\newrefformat{proto}{\savehyperref{#1}{Protocol~\ref*{#1}}}
\newrefformat{prob}{\savehyperref{#1}{Problem~\ref*{#1}}}
\newrefformat{claim}{\savehyperref{#1}{Claim~\ref*{#1}}}
\newrefformat{op}{\savehyperref{#1}{Open Problem~\ref*{#1}}}

\usepackage[%
linewidth=2pt,
linecolor=gray,
middlelinecolor= black,
middlelinewidth=0.4pt,
roundcorner=1pt,
topline = false,
rightline = false,
bottomline = false,
rightmargin=0pt,
skipabove=0pt,
skipbelow=0pt,
leftmargin=0pt,
innerleftmargin=4pt,
innerrightmargin=0pt,
innertopmargin=0pt,
innerbottommargin=0pt,
]{mdframed}

\newenvironment{nproblem}[1][\unskip]{%
\medskip
\begin{mdframed}
\noindent
\textbf{\underline{$#1$ Problem.}} \\
\noindent}
{\end{mdframed}}

\usepackage[scaled=.90]{helvet}
\usepackage{xspace}
\usepackage{complexity}

\newcommand{\gdmax}{\textsc{GradientDescentMax}}

\newcommand{\CLS}{\ComplexityFont{CLS}}

\newcommand{\FIXP}{\ComplexityFont{FIXP}}
\let\vec\undefined
\newcommand{\abs}[1]{|#1|}
\newcommand{\vec}[1]{\bm{#1}}
\newcommand{\norm}[1]{\left\| #1 \right\|}

\newcommand{\tmax}{\textup{\textbf{MAXMIN}}}
\newcommand{\tmin}{\textup{\textbf{MINMAX}}}
\newcommand{\tlocmax}{\textup{\textbf{LocMAXMIN}}}
\newcommand{\tlocmin}{\textup{\textbf{LocMINMAX}}}
\newcommand{\calA}{\ensuremath{\mathcal{A}}}
\newcommand{\calB}{\ensuremath{\mathcal{B}}}
\newcommand{\calC}{\ensuremath{\mathcal{C}}}

\newcommand{\calE}{\ensuremath{\mathcal{E}}}
\newcommand{\calF}{\ensuremath{\mathcal{F}}}

\newcommand{\calM}{\ensuremath{\mathcal{M}}}
\newcommand{\calN}{\ensuremath{\mathcal{N}}}

\newcommand{\calR}{\ensuremath{\mathcal{R}}}

\newcommand{\calX}{\ensuremath{\mathcal{X}}}
\newcommand{\calY}{\ensuremath{\mathcal{Y}}}

\renewcommand{\R}{\mathbb{R}}
\newcommand{\N}{\mathbb{N}}
\renewcommand{\E}{\mathbb{E}}
\newcommand{\defeq}{\coloneqq}
\newcommand{\prox}{\mathrm{prox}}

\usepackage{cleveref}

\Crefname{conjecture}{Conjecture}{Conjectures}
\crefname{assumption}{assumption}{assumptions}%
\crefname{question}{question}{questions}

\newtheorem*{theorem*}{Theorem}
\newtheorem*{proposition*}{Proposition}

\newcommand{\vx}{\vec{x}}
\newcommand{\vy}{\vec{y}}
\newcommand{\ay}{\vec{y}'}
\newcommand{\vz}{\vec{z}}

\newcommand{\opt}{\textrm{opt}}

\usepackage{authblk}

\author[1]{Ioannis Anagnostides}
\author[2]{Fivos Kalogiannis}
\author[2]{Ioannis Panageas}
\author[3]{Emmanouil-Vasileios Vlatakis-Gkaragkounis}
\author[1]{Stephen McAleer}

\affil[1]{Carnegie Mellon University}
\affil[2]{University of California, Irvine}
\affil[3]{University of California, Berkeley}

\title{Algorithms and Complexity for Computing Nash Equilibria in Adversarial Team Games\footnote{Correspondence to \texttt{ipanagea@ics.uci.edu}.}}

\begin{document}
\date{}
\maketitle
\pagenumbering{gobble}
\begin{abstract}
    \emph{Adversarial team games} model multiplayer strategic interactions in which a team of identically-interested players is competing against an adversarial player in a zero-sum game. Such games capture many well-studied settings in game theory, such as congestion games, but go well-beyond to environments wherein the cooperation of one team---in the absence of explicit communication---is obstructed by competing entities; the latter setting remains poorly understood despite its numerous applications. Since the seminal work of Von Stengel and Koller (GEB `97), different solution concepts have received attention from an algorithmic standpoint. Yet, the complexity of the standard Nash equilibrium has remained open.

    In this paper, we settle this question by showing that computing a Nash equilibrium in adversarial team games belongs to the class \emph{continuous local search ($\CLS$)}, thereby establishing $\CLS$-completeness by virtue of the recent $\CLS$-hardness result of Rubinstein and Babichenko (STOC `21) in potential games. To do so, we leverage linear programming duality to prove that any $\epsilon$-approximate stationary strategy for the team can be extended in polynomial time to an $O(\epsilon)$-approximate Nash equilibrium, where the $O(\cdot)$ notation suppresses polynomial factors. As a consequence, we show that the \emph{Moreau envelope} of a suitable best response function acts as a \emph{potential} under certain natural gradient-based dynamics.

\end{abstract}

\clearpage
\pagenumbering{arabic}
\section{Introduction}

Computing \emph{Nash equilibria} has served as one of the most influential problems in the development of algorithmic game theory, with a myriad of applications in multi-agent systems. In the special case of \emph{two-player zero-sum} games~\citep{Von07:Theory}, computing a Nash equilibrium is known to be equivalent to linear programming, thereby admitting efficient algorithms~\citep{Adler13:The}. In stark contrast, the problem becomes computationally intractable in general games~\citep{Daskalakis09:The,Chen09:Settling}, a hardness that persists even under approximate equilibrium concepts~\citep{Rubinstein16:Settling,Daskalakis13:On}. As a result, a central research endeavor is to identify and characterize classes of games that elude the aforedescribed intractability barriers.

In this paper, we study \emph{adversarial team games}, a fundamental and well-studied setting in economics in which a team of identically-interested players is competing against an adversary in a zero-sum game. 
Such games capture many important applications, including competition between firms, decision-making in distributed computer systems, diplomacy between political entities, and biological systems; \emph{e.g.}, see \citep{Coase37:The,Marschak55:Elements,Tomoko02:Near,SchulmanV19}. Furthermore, that class incorporates as special cases two-player zero-sum games and \emph{potential games}, two of the most well-studied classes of games in the literature, but goes well-beyond to settings that feature both competing and shared interests. In those environments, it is understood that if players have the capacity to communicate without any restrictions, the game reduces to a standard two-player zero-sum interaction. However, coordination between the members of the team can be difficult or even infeasible to achieve in many applications, an impediment recognized as one of the major challenges in group decision-making~\citep{Marschak55:Elements,VonStengel97:Team,SchulmanV19}. \textcolor{black}{ Indeed, the lack of coordination has long been understood as one of the key driving forces in establishing the composition and size of a firm~\citep{Coase37:The}. Beyond the ubiquitous presence of such applications in the economics literature, further examples arise in biology concerning the dynamics of evolving species under the so-called \emph{weak selection} model~\citep{Chastain14:Algorithms,Ohta02:Near}.}

From an algorithmic standpoint, commencing from the pioneering work of~\citet{VonStengel97:Team}, the complexity of different equilibrium concepts in adversarial team games has received extensive attention in recent years (see our overview in \Cref{sec:related}). Yet, surprisingly, the complexity of the standard Nash equilibrium has eluded prior investigations, bringing us to our central question: 

\begin{question}
    \label{qes:Nash}
    What is the complexity of computing Nash equilibria in adversarial team games?
\end{question}

\subsection{Our Contributions}

Our primary contribution is to establish that computing an approximate Nash equilibrium in adversarial team games lies in \emph{continuous local search (\CLS)}, a complexity class introduced by~\citet{Daskalakis11:Continuous} that has received extensive attention in recent years~\citep{Hubacek20:Hardness,Fearnley20:Unique,Goos18:Adventures,Daskalakis18:Converse,Fearnley17:Cls,Gartner18:Arrival,Goos22:Further,Goos22:Seperations,Tewolde23:Computational}, culminating in the recent breakthrough result of~\citet{Fearnley21:The}. \textcolor{black}{Thus, when combined with the recent $\CLS$-hardness result of~\citet{babichenko2021settling} for identical-interest games---a special case of adversarial team games in which the adversary has no effect on the game}, we settle \Cref{qes:Nash}:

\begin{theorem*}
    Computing an $\epsilon$-approximate Nash equilibrium in adversarial team games is $\CLS$-complete.
\end{theorem*}

We remark that, as is standard, the above theorem applies in the regime where $\log(1/\epsilon)$ is a polynomial; we do not address the complexity of computing an \emph{exact} Nash equilibrium, which is left as an interesting open problem (see \Cref{conj:FIXP}).

Given that computing Nash equilibria is well-known to be in $\PPAD$~\citep{Daskalakis09:The,papadimitriou1994complexity}, membership in $\CLS$ reduces to showing that the problem is in the class \emph{polynomial local search ($\PLS$)}~\citep{johnson1988easy} by virtue of the recent collapse $\CLS = \PPAD \cap \PLS$~\citep{Fearnley21:The}. To do so, the starting point is the fundamental \emph{extendibility} of team-maxmin equilibria shown by~\citet{VonStengel97:Team}; namely, \emph{any team-maxmin profile can be extended to an equilibrium of the game}. Unfortunately, computing even an approximate team-maxmin profile is computationally prohibitive~\citep{Borgs10:The,Hansen08:Approximability},\footnote{\textcolor{black}{\citet{Borgs10:The} did not phrase their hardness result explicitly in terms of computing a team-maxmin equilbirum, but there is an immediate equivalence between computing the so-called \emph{threat point} in the context of the folk theorem and team-maxmin equilibria.} } being $\FNP$-hard~\citep{Celli18:Computational}, thereby necessitating a more refined approach.

In light of this, one of our main insights is to relax the requirement of the extendibility result of~\citet{VonStengel97:Team}. To be more precise, let us introduce some basic notation about our setting; the formal description is deferred to~\Cref{sec:prels}. For an adversarial team game $\Gamma$, we let $\calX$ be the joint strategy space of the team and $\calY$ be the strategy space of the adversary. We further suppose that $U : \calX \times \calY \to \R$ is the (mixed extension) of the utility function. We will also write $\poly(\Gamma)$ for factors that are polynomial in the natural parameters of the game. In this context, we show the following central characterization.

\begin{proposition*}[Formal Version in \Cref{thm:onegeneralize}]
    Let $\vx^* \in \calX $ be an $\epsilon$-near stationary point of the function $\vx \mapsto \max_{\vy \in \calY} U(\vx, \vy)$. Then, there exists $\vy^* \in \calY$ such that $(\vx^*, \vy^*)$ is an $(\epsilon \cdot \poly(\Gamma))$-approximate Nash equilibrium of $\Gamma$. Further, $\vy^*$ can be computed in polynomial time.
\end{proposition*}

This characterization is established via linear programming duality, similarly to the result of~\citet{VonStengel97:Team}. However, unlike the extension of~\citet{VonStengel97:Team}, the proposition above is based on an approximate stationary point $\vx^*$, which, crucially, can be efficiently computed.\footnote{\textcolor{black}{ It is worth noting here that the proof of~\citet{VonStengel97:Team} works more generally for any strategy profile that is unimprovable through local deviations with respect to \emph{the entire strategy profile}, as pointed out explicitly by the authors. However, it is also mentioned that the proof does not carry over when one instead considers \emph{unilateral} deviations. It is unclear how one can use the local notion of~\citet{VonStengel97:Team}, as described earlier, for computational purposes since it is in general \NP-hard to compute---the natural approach consisting of discretizing the domain and applying some ``local'' method is problematic because the natural neighborhood now has a prohibitively large size. On the other hand, and crucially for our paper, computing a stationary point of a differentiable function that is in \CLS.} } Indeed, we analyze a natural gradient-based algorithm (\Cref{alg:gdmax}), which additionally incorporates a suitable linear program that performs the extension from $\vx^*$ to $(\vx^*, \vy^*)$. Importantly, we show in \Cref{theorem:pot} that there exists a polynomially computable potential function---namely, the approximate value of the \emph{Moreau envelope} (\Cref{def:moreau}) of $\vx \mapsto \max_{\vy \in \calY} U(\vx, \vy)$. Along with the fact that each iteration of \Cref{alg:gdmax} can be implemented in polynomial time, this suffices for showing membership in $\PLS$, modulo some further technical issues that we address. \textcolor{black}{As a byproduct, our algorithm also leads to a fully polynomial-time approximation scheme for computing $\epsilon$-approximate Nash equilibria in adversarial team games (\Cref{cor:fptas}), under the standard assumption that utilities lie in $[0, 1]$. Consequently, from an economic standpoint, a key message of our work is that---under the complexity landscape as currently understood---computing Nash equilibria in adversarial team games is easier than in general games. And perhaps more importantly, Nash equilibria in adversarial team games arise as the limit point of natural gradient-based learning dynamics.}

From a broader viewpoint, we find it surprising that computing Nash equilibria in adversarial team games has eluded prior research, although it subsumes many settings that have received tremendous interest in algorithmic game theory, such as congestion games. But, importantly, our setting also covers more uncharted territory; one interesting such example, discussed in \Cref{sec:advcong}, is that of congestion games with \emph{adversarial costs} (\emph{cf}. \citet{Babaioff09:Congestion}). We hope that our results will encourage further research in those directions. Indeed, as we highlight in \Cref{sec:conclusion}, there are still many exciting open problems. Perhaps most notably, while our paper focuses on the canonical setting where the team is facing a single adversary, we leave as a challenging open question the complexity of the more general problem in which both teams have multiple players~\citep{SchulmanV19}.

\subsection{Further Related Work}
\label{sec:related}

The study of team games from an algorithm standpoint dates back at least to the pioneering work of~\citet{VonStengel97:Team}. In particular, they introduced the \emph{team-maxmin equilibrium (TME)} for games wherein a team of players---with identical payoffs---faces a single adversary. A TME profile prescribes a mixed strategy for each team member so that the minimal expected team payoff over all possible responses of the adversary (potentially knowing the play of the team) is the maximum possible. That is, in a TME profile the players of the team guarantee the best possible payoff, given the lack of coordination. Naturally, if enough communication is added between the team members, they would be able to coordinate all their actions, and the team-maxmin payoff would be identical to the value of the game when viewed as a two-player zero-sum game. However, prior work~\citep{VonStengel97:Team,SchulmanV19} has provided ample motivation for studying team games with no coordination or communication. In other words, each player can only use a separate mixed (randomized) strategy, but correlated randomization (beyond the structure of the game) is not allowed.

TME enjoy a number of desirable properties. For one, a team-maxmin equilibrium has a \emph{unique value}, a compelling antidote to the equilibrium selection problem in general games. Further, that value is the optimal payoff that individual players can guarantee: no equilibrium of the game can be better for the team than a team-maxmin equilibrium~\citep{VonStengel97:Team}, while the worst equilibrium can be far from a TME in terms of efficiency~\citep{Basilico17:Team}. Unfortunately, computing a TME suffers from computational intractability, being \FNP-hard~\citep{Hansen08:Approximability,Borgs10:The}. Nevertheless, several methods have been designed to tackle the problem in practice~\citep{Converging20:Zhang,Zhang20:Computing,Basilico17:Team}. 

A different avenue for addressing the intractability of TME that has received attention in recent years is to incorporate some correlation device, thereby allowing players to coordinate \emph{ex ante}~\citep{Zhang21:Computing,Farina18:Ex,Farina21:Connecting,Zhang22:Team,Zhang21:Team,Celli18:Computational,Celli19:Coordination,Carminati22:A}. Nevertheless, even with a correlation device, the problem is $\NP$-hard~\citep{Chu01:On}, although it has been shown that efficient algorithms exist as long as the asymmetry between the players' information is limited~\citep{Zhang21:Computing}. Further, it is worth noting that team equilibria are also useful for extensive-form two-player zero-sum games where one of the players has \emph{imperfect recall}~\citep{Kaneko95:Behavior,Piccione97:On}. 

Moreover, in a subsequent work, \citet{Kalogiannis22:Efficiently} extended our techniques to obtain an $\FPTAS$ for computing Nash equilibria in the more general class of adversarial team \emph{Markov} games. We also refer to the works of~\citet{Kartik22:Common} and~\citet{Lagoudakis02:Learning}, as well as the many references therein, for other pertinent considerations. Finally, we note that team games have also been recently analyzed from a mean-field perspective~\citep{Guan23:Zero}.

\section{Preliminaries}
\label{sec:prels}

In this section, we introduce the necessary background on team games and optimization theory. 

\subsection{Two-Team Zero-Sum Games} 
\label{sec:teams}

A \emph{two-team zero-sum game},\footnote{While our main focus is on the case where the adversary team has a single player, commonly referred to as adversarial team games in the literature, we present the more general setting here so as to smoothly discuss future directions in \Cref{sec:advcong,sec:conclusion}.} represented in normal form, is defined by a tuple $\Gamma(\calN, \calM,\calA,
\calB,U)$.\footnote{For notational convenience, we will sometimes use the notation $\Gamma(n, U)$ to denote an adversarial team game with $n$ players (minimizers) in team $A$ and a single adversary ($m = 1$) .} $\Gamma$ consists of a finite set of $n = |\calN|$ \emph{players} belonging to team $A$, and a finite set of $m = |\calM|$ \emph{players} belonging to team $B$. Each player from team $A$ has a finite and nonempty set of available \emph{actions} $\calA_i$, so that $\calA \defeq \prod_{i=1}^n \calA_i$ denotes the ensemble of all possible action profiles of team $A$. Similarly, each player from team $B$ has a finite and nonempty set of actions $\calB_j$ per player $j \in \calM$. We will denote by $\vec{a}=(a_1,\ldots, a_n) \in \calA$ the action profile of team $A$, and $\vec{b}=(b_1,\ldots,b_m) \in \calB \defeq \prod_{j=1}^m \calB_j$ the action profile of team $B$. Each team's \emph{payoff} function is denoted by $U_{A},U_{B}:\calA\times\calB\to\R$, so that the \emph{individual utility} of a player is identical to its teammates: $U_{i}(\vec{a}, \vec{b}) = U_{A}(\vec{a}, \vec{b})$ and $U_{j}(\vec{a}, \vec{b}) = U_{B}(\vec{a}, \vec{b})$, for all joint action profiles $(\vec{a}, \vec{b}) \in \calA \times \calB$ and for all players $i \in \calN$ and $j  \in \calM$. Further, the team game is assumed to be zero-sum, in the sense that $U_{B}(\vec{a}, \vec{b}) =- U_{A}(\vec{a}, \vec{b}) = U(\vec{a}, \vec{b})$. As a result, players in team $B$ aim to maximize $U$---thereby referred to as maximizers, while players in team $A$ aim to minimize $U$ (hereinafter, minimizers).

To ensure existence of equilibria, players are allowed to randomize, that is, select a probability distribution over their set of actions. We define the product distributions $\vec{x} = (\vec{x}_1,\ldots, \vec{x}_n)$, $\vec{y} = (\vec{y}_1, \ldots, \vec{y}_m)$ as the joint strategies of teams $A$ and $B$, respectively, so that $\vec{x}_i \in \Delta(\calA_i)$ and $\vec{y}_j \in \Delta(\calB_j)$, for any $i \in \calN$ and $j \in \calM$. For convenience, we will write $\calX \defeq \prod_{i\in\calN}\calX_i \defeq \prod_{i\in\calN}\Delta(\calA_i)$ and $\calY \defeq \prod_{j \in \calM}\calY_j \defeq \prod_{j \in \calM} \Delta(\calB_j)$ for the space of mixed strategy profiles of teams $A$ and $B$ respectively. Finally, we overload notation so that $U : \calX \times \calY \ni (\vec{x}, \vec{y}) \mapsto \E_{(\vec{a}, \vec{b}) \sim (\vx, \vy)} U(\vec{a}, \vec{b})$.

\paragraph{Nash equilibrium.} In terms of solution concepts, our main focus is on computing \emph{$\epsilon$-approximate Nash Equilibria (NE)}, that is, a strategy profile $(\vec{x}^*,\vec{y}^*) \in \calX\times \calY$ such that for any possible unilateral deviation $\vx_i \in \calX_i$ from any player $i \in \calN$ or $\vy_j \in \calY_j$ from any player $j \in \calM$,
\begin{equation}
    \label{def:approximatenash}
\begin{array}{cc}
    U(\vec{x}^*,\vec{y}^*)\le U(\vec{x}_{i},\vec{x}^*_{-i},\vec{y}^*) +\epsilon \textrm{ and } U(\vec{x}^*,\vec{y}^*)\ge U(\vec{x}^*,\vec{y}_{j},\vec{y}^*_{-j}) - \epsilon.
    \tag{NE}\end{array}
\end{equation}
Here, we used the standard shorthand notation $\vec{x}_{-i} = (\vec{x}_1, \dots, \vec{x}_{i-1}, \vec{x}_{i+1}, \vec{x}_n) \in \prod_{i' \neq i} \Delta(\calA_{i'})$, and similarly for $\vec{y}_{-j} \in \prod_{j' \neq j} \Delta(\calB_{j'})$. The existence of such an equilibrium point is guaranteed by Nash's theorem~\citep{Nash51:Non}. We will say that the strategy profile $(\vec{x}^*,\vec{y}^*)$ is \emph{pure} if each player is selecting an action with probability $1$. 

\paragraph{Team-maxmin equilibrium.} Another natural equilibrium concept tailored to two-team zero-sum games is the \emph{team-maxmin equilibrium (TME)}, \textcolor{black}{which induces a Nash equilibrium in adversarial team games~\citep{VonStengel97:Team}.} \textcolor{black}{A joint strategy profile $(\vec{x}^*,\vec{y}^*) \in \calX \times \calY$ is a team-maxmin equilibrium if}
 \begin{equation}\label{eq:maxmin} 
 \vec{y}^* \in \arg\max_{\vec{y}\in \calY} \min_{\vec{x} \in \calX} U(\vec{x},\vec{y}) \textrm{ and } \vec{x}^* \in \arg\min_{\vec{x}\in \calX} U(\vec{x},\vec{y}^*)  \tag{Maxmin}.
 \end{equation}
 Any team-maxmin equilibrium yields the same value~\citep{VonStengel97:Team}, denoted by $U_{\textrm{maxmin}}$. Similarly, a team-minmax equilibrium $(\vec{x}^*,\vec{y}^*) \in \calX \times \calY$ satisfies
 \begin{equation}\label{eq:minmax} 
 \vec{x}^* \in \arg \min_{\vec{x} \in \calX} \max_{\vec{y} \in \calY} U(\vec{x},\vec{y}) \textrm{ and } \vec{y}^* \in \arg\max_{\vec{y}\in \calY} U(\vec{x}^*,\vec{y}). \tag{Minmax}
 \end{equation}
 
 Any team-minmax equilibrium yields the same value, namely $U_{\textrm{minmax}}$. However, unlike two-player zero-sum games, team games do not enjoy a minimax theorem. That is, while $U_{\textrm{minmax}} \geq U_{\textrm{maxmin}}$ (by weak duality), equality does not hold in general~\citep{SchulmanV19}.

\begin{remark}[Succinct representation]\label{rem:succint}
A standard issue encountered in multiplayer games is that the utility tensors describing the game have size that scales exponentially with the number of players. As is common, we will bypass this issue by focusing on multiplayer two-team zero-sum games that have the \emph{polynomial expectation property}~\citep{Papadimitriou08:Computing}, formalized below (\Cref{assumption:polexp}). This property is known to hold for most classes of succinctly representable games~\citep{Papadimitriou08:Computing}, albeit not for all~\citep{Daskalakis06:The}. \textcolor{black}{Formally, a succinct game $\Gamma$ is given implicitly through the use of an input $\mathcal{I}$. We assume throughout that the game is of \emph{polynomial type}, meaning that both the number of players $n$ and the action set of each player are polynomial in the size of the input $\mathcal{I}$; numerous such examples are given by~\citet{Papadimitriou08:Computing}. Naturally, when we say that an algorithm is polynomial, it is with respect to the size of the input $\mathcal{I}$, the succinct representation of $\Gamma$. Finally, we assume throughout that $| \mathcal{I} | \leq \poly(n, \sum_{i=1}^n |\calA_i|, |\calB|)$.}

\end{remark}

To state the next assumption below, we will use the notation $|\vec{x}|$ to represent the number of bits required to represent vector $\vec{x}$.

\begin{assumption}[Polynomial Expectation Property]
    \label{assumption:polexp}
    \textcolor{black}{Let $\Gamma$ be a game defined succinctly through input $\mathcal{I}$. For any (mixed) joint strategy profile $(\vec{x}, \vec{y}) \in \prod_{i=1}^n \Delta(\calA_i) \times \prod_{j=1}^m \Delta(\calB_j)$, we can compute (exactly) the expectation $\E_{(\vec{a}, \vec{b}) \sim (\vec{x}, \vec{y})} U(\vec{a}, \vec{b}) \eqqcolon U(\vx, \vy)$ in time $\poly(|\mathcal{I}|, |\vec{x}|, |\vec{y}|)$.}
\end{assumption}

\subsection{Basic Background on Optimization} 

First, it will be useful to note that any point $(\vec{x}^*,\vec{y}^*) \in \calX\times\calY$ that is a solution to the following \emph{variational inequality (VI)} problem
 \begin{equation}\label{def:vi}    \nabla_{\vec{x}} U(\vec{x}^*,\vec{y}^*)(\vec{x}^*-\vec{x})\leq \epsilon \; \forall \vec{x} \in \calX \textrm{ and }\nabla_{\vec{y}} U(\vec{x}^*,\vec{y}^*)(\vec{y}^*-\vec{y})\geq -\epsilon \; \forall \vec{y} \in \calY, \tag{VI}
 \end{equation}
must also be an $\epsilon$-approximate NE of the corresponding game. Further, let us define $\phi: \vx \mapsto \max_{\vy \in \calY} U(\vx, \vy)$, which is the utility function of the team when the adversary is best responding.  The following is a well-known property about $\phi$~\citep{Davis18:Stochastic}.

\begin{lemma}
    \label{lemma:Lip}
    If $U(\vx, \vy)$ is $L$-Lipschitz continuous, then the function $\vx \mapsto \max_{\vy \in \calY} U(\vx, \vy)$ is also $L$-Lipschitz continuous.
\end{lemma}

On the other hand, $\phi$ is \emph{not smooth}, and so standard techniques in optimization no longer apply when using gradient descent on $\phi$. \textcolor{black}{To construct a potential function, we will instead rely on some recent developments concerning gradient descent on nonsmooth functions~\citep{lin2020gradient}.} To do so, we make use of the \emph{Moreau envelope}, introduced below; this will be crucial for our potential argument in \Cref{sec:cls}.

\begin{definition}[Moreau Envelope]
    \label{def:moreau}
    Consider any game $\Gamma$. We define the max-Moreau envelope with parameter $\lambda > 0$ as
\begin{equation*}
\phi_{\lambda}(\vec{x}) \textcolor{black}{\defeq \min_{\vec{x}'\in\R^{\sum_i |\calA_i|}} \left\{ \max_{\vec{y}\in\calY} U(\vec{x}',\vec{y}) + r(\vec{x}')+\frac{1}{2\lambda} \norm{\vec{x}-\vec{x}'}_2^2 \right\} }= \min_{\vec{x}'\in\calX} \left\{ \max_{\vec{y}\in\calY} U(\vec{x}',\vec{y}) + \frac{1}{2\lambda} \norm{\vec{x}-\vec{x}'}_2^2 \right\},
\end{equation*}
\textcolor{black}{where $r(\vec{x}') : \R^{\sum_i |\calA_i|} \to \R \cup \{\infty\}$ is zero if $\vec{x}'\in\calX$ and infinity otherwise (convex function).} 
\end{definition}
It is well-known that for a sufficiently small $\lambda$, $\phi_{\lambda}$ defined above is a differentiable function in ; \emph{e.g.}, see~\citep{Davis18:Stochastic}. In particular, if $\ell$ is a smoothness parameter of $U$, it suffices to take $\lambda \defeq \frac{1}{2\ell}$; \textcolor{black}{we recall that a continuously differentiable function $U$ is called \emph{$\ell$-smooth} if $\| \nabla U(\vx, \vy) - \nabla U(\vx', \vy') \|_2 \leq \ell \| (\vx, \vy) - (\vx', \vy') \|_2$ for any $\vx, \vx' \in \calX$ and $\vy, \vy' \in \calY$. The utility function $U$ in the context of (normal-form) games is guaranteed to be Lipschitz continuous and smooth by virtue of multilinearity (see~\Cref{lemma:Lipcon}).} We will also need the following well-known property~\citep{Davis18:Stochastic,lin2020gradient}.

\begin{lemma}
    \label{lemma:weaklyconvex}
    If $U(\vx, \vy)$ is $\ell$-smooth, then the function $\vx \mapsto \max_{\vy \in \calY} U(\vx, \vy)$ is $\ell$-weakly convex, in the sense that the function $\max_{\vy \in \calY} U(\vx, \vy) + \frac{\ell}{2}\|\vx\|_2^2$ is convex.
\end{lemma}

\textcolor{black}{In particular, this implies that $\phi(\vx) + \ell \|\vx\|^2_2$ is \emph{$\ell$-strongly convex} since $\frac{\ell}{2} \| \vx \|_2^2$ is $\ell$-strongly convex (by definition), a fact that we will use in the sequel.} Finally, we will need the definition of a \emph{stationary point} for the nonsmooth function $\phi: \vx \mapsto \max_{\vy \in \calY} U(\vx, \vy)$, formally introduced below through the use of the Moreau envelope.

\begin{definition}[Approximate Stationarity]
    \label{def:stat}
    We call a point $\vec{x}^* \in \calX$ an $\epsilon$-approximate stationary point of $\phi_{1/(2\ell)}$ if
 \begin{equation}
 \label{eq:stationary1} 
\textcolor{black}{\norm{\nabla \phi_{1/(2\ell)}(\vec{x}^*)}_2 \leq \epsilon}.
 \end{equation}
\end{definition}

A point $\vx^* \in \calX$ that satisfies~\eqref{eq:stationary1} will also sometimes be referred to as \emph{$\epsilon$-near stationary} point of $\phi$. The next lemma connects the closeness of a point to stationarity to the distance from its \emph{proximal point}. More precisely, we recall that the proximal point of $\vx^* \in \calX$ is defined as $\hat{\vx} \defeq \arg \min_{\vx \in \calX} \left\{ \phi(\vx) + \ell \|\vx - \vx^*\|_2^2 \right\}$, and it is indeed well-defined~\citep{lin2020gradient}.

\begin{lemma}
    \label{lemma:close-prox}
    Suppose that $\vx^* \in \calX$ is an $\epsilon$-approximate stationary point of $\phi_{1/(2\ell)}$. If
    \begin{equation*}
    \calX \ni \hat{\vx} \defeq \arg \min_{\vx \in \calX} \left\{ \phi(\vx) + \ell \|\vx - \vx^*\|_2^2 \right\}
    \end{equation*}
    is the proximal point of $\vx^*$, then \textcolor{black}{$\|\vx^* - \hat{\vx}\|_2 = \frac{\| 
\nabla \phi_{1/(2\ell)} (\vx^*) \|_2}{2\ell}\leq \frac{\epsilon}{2\ell}$}. \textcolor{black}{Further, there exists a $\vec{\xi} \in\partial \phi (\hat{\vx})$ so that
 $\norm{\vec{\xi}}_2 \leq \epsilon.$}
\end{lemma}

\noindent\textcolor{black}{For this lemma see in the works of~\citep[end of page 2]{Davis18:Stochastic}, ~\citep{lin2020gradient},~\citep[Lemma 6]{Daskalakis20:Independent} and references therein. In what follows, we will also need the following lemma.}

\begin{lemma}
    \label{lemma:Golowich}
    Let $f : \calX \times \calY \to \R$ be $L$-Lipschitz continuous and $\ell$-smooth such that the function $\vx \mapsto f(\vx, \vy)$ is linear for any $\vy \in \calY$. If $\phi : \vx \mapsto \max_{\vy \in \calY} f(\vx, \vy)$, then for all $\vx \in \calX$
    \[
        \phi(\vx) - \min_{\vx^* \in \calX} \phi(\vx^*) \leq \left( 1 + \frac{L}{2\ell} \right) \| 
\nabla \phi_{1/(2\ell)} (\vx) \|_2.
    \]
\end{lemma}
\textcolor{black}{
A more general version of this lemma appears as~\citep[Lemma A.2]{Daskalakis20:Independent}.
}

\section{The Complexity of Adversarial Team Games}
\label{sec:cls}

In this section, we 
establish our main result: computing $\epsilon$-approximate Nash equilibria in adversarial team games is complete for the class $\CLS$, as formalized in \Cref{theorem:cls}. We organize our argument as follows.
\begin{itemize}
    \item First, in \Cref{sec:dualLP} we show that an $\epsilon$-approximate stationary point $\vx^* \in \calX$ of the Moreau envelope $\phi_{1/(2\ell)}$ \emph{can be extended to an $O(\epsilon)$-approximate Nash equilibrium $(\vx^*, \vy^*) \in \calY$} (\Cref{thm:onegeneralize}). This extends a result of~\citet{VonStengel97:Team}, and it is established via strong duality.
    \item Next, in \Cref{sec:potential} we construct an explicit polynomially computable potential function (\Cref{theorem:pot}) using gradient descent on a suitable function, along with the extendibility argument of \Cref{thm:onegeneralize}.
\end{itemize}

\subsection{From a Stationary Point to a Nash Equilibrium via Strong Duality}
\label{sec:dualLP}

Suppose that $\vx^* \in \calX$ is an $\epsilon$-approximate stationary point of $\phi_{1/(2\ell)}$. We will show (in \Cref{thm:onegeneralize}) that $\vx^*$ can be extended to an $O(\epsilon)$-approximate Nash equilibrium profile. A similar in spirit extendibility argument was known from the work of~\citet{VonStengel97:Team}, but that required that $\vx^*$ is a team-maxmin profile, which is computationally prohibitive for establishing membership in $\CLS$. On the other hand, as we formalize in \Cref{sec:potential}, approximate stationary points of $\phi_{1/(2\ell)}$ can be computed in time $\poly(\Gamma, 1/\epsilon)$, where $\poly(\Gamma) \defeq \poly(\sum_{i=1}^n | \calA_i |, |\calB|, V)$; here, $V > 0$ denotes the maximum absolute value of a utility in the payoff tensor.

Let $L = \poly(\Gamma)$ be the Lipschitz parameter of $U(\vx, \vy)$ (\Cref{lemma:Lipcon}), and $\ell = \poly(\Gamma)$ be the smoothness parameter of $U(\vx, \vy)$ (analogously to \Cref{lemma:Lipcon}). The theorem below is the main result of this subsection. Our extendibility argument leverages strong duality, as in~\citep{VonStengel97:Team}, but, crucially, our argument only requires $\vx^* \in \calX$ to be an approximate stationary point, instead of a team-maxmin profile.


\begin{theorem}[Extending an Approximate Stationary Point]\label{thm:onegeneralize} Given an adversarial team game $\Gamma(n, U)$ we let $\vec{x}^* \in \calX$ be a sufficiently small $(\epsilon/\poly(\Gamma))$-near stationary point of the function $ \vx \mapsto \max_{\vec{y}\in \calY} U(\vec{x},\vec{y})$. Then, there exists a strategy for the adversary $\vec{y}^* \in \calY$ so that $(\vec{x}^*,\vec{y}^*)$ is an
\begin{equation}
    \label{eq:approxNE1}
    \epsilon-\textrm{approximate Nash equilibrium}.
\end{equation}
Furthermore, under \Cref{assumption:polexp}, such a strategy $\vy^* \in \calY$ can be computed in polynomial time by solving a suitable linear program.
\end{theorem}

In what follows, we will use the notation $O(\cdot)$ to suppress $\poly(\Gamma)$ factors in order to simplify the exposition.

\begin{proof}[Proof of \Cref{thm:onegeneralize}] First, by linearity, we can write $U(\vx,\vy) := \sum_{b \in \calB} y(b) U_b(\vec{x})$, for some function $U_b : \calX \to \R$, for any action $b \in \calB$. We consider the following linear program with variable $u \in \R$:
\begin{equation}
    \label{eq:flp}
\begin{aligned}
    &\mathrlap{\min~ u} \\ 
    &\textrm{ s.t. } & u \geq U_b(\vx^*) & \textrm{ for all }b\in\calB,
\end{aligned}    
\end{equation}
where we recall that $\vx^*$ is an $\epsilon$-near stationary point of $\max_{\vec{y}\in\calY} U(\vec{x},\vec{y})$. Furthermore, we construct the following linear program with variables $(u,\vx) \in \R \times \calX$:
\begin{equation}
    \label{eq:primalLP}
\begin{aligned}
    &\mathrlap{\min ~ n u} \\
    &\textrm{s.t.} &nu \geq \sum_{i=1}^n U_b(\vx_i,{\vx}^*_{-i}) & \textrm{ for all }b\in\calB,\\
    & &\vec{x}_i \in \Delta(\calA_i) & \textrm{ for all }i \in \calN.
\end{aligned}
\end{equation}

Now let $u^* \in \R$ be the minimum of the first program~\eqref{eq:flp}; the existence of such $u^*$ is evident from the constraints of \eqref{eq:flp}. Further, let us assume that the second program attains a minimum at $(u^{\opt}, \vx^{\opt})$. We observe that the point $(u^*,\vx^*) \in \R \times \calX$ is feasible for the second LP since, by definition, $\vx^* \in \prod_{i=1}^n \Delta(\calA_i)$ belongs to the product of simplices, and $nu^* \geq \sum_{i=1}^n U_b(\vx^*) = n U_b(\vx^*)$, by feasibility of $u^*$ in~\eqref{eq:flp}. Therefore, we conclude that 
\begin{equation}
    \label{eq:star}
    u^* \geq u^{\opt}.
\end{equation}

Let $\hat{\vx} \in \calX$ be the proximal point of $\vx^*$ for the function $\max_{y\in \calY} U(\vx,\vy)$. By \Cref{lemma:close-prox}, it follows that $\norm{\hat{\vx}-\vx^*}_2 \leq \frac{\epsilon}{2\ell}$. Further, given that $u^*$ is the optimal solution to~\eqref{eq:flp}, it follows that $u^* = \max_{b \in \calB} U_b(\vx^*) = \max_{\vy \in \calY} U(\vx^*, \vy)$, in turn implying that
\begin{equation}\label{eq:cont}
nu^* = \max_{\vy \in \calY} \sum_{i=1}^n U(\vx^*,\vy) \leq \max_{\vy \in \calY} \sum_{i=1}^n U(\hat{\vx},\vy) + n L \frac{\epsilon}{2\ell},
\end{equation}
where we used the $L$-Lipschitz continuity of $\max_{\vy\in \calY} U(\vx,\vy)$ (\Cref{lemma:Lip}). 
\noindent \textcolor{black}{We also consider the function $w(\vec{x})\defeq \max_{\vy\in\calY} \sum_{i=1}^n U(\vec{x}_i,\hat{\vec{x}}_{-i},\vec{y}).$ We invoke \Cref{lem:moreaulinear} and get that the Moreau envelope of $w$ computed at $\hat{\vec{x}}$ gives} 
\begin{equation}\label{eq:smallw}
\norm{  \nabla w_{1/(2\ell)}(\hat{\vx}) }_2 \leq 4 \epsilon.
\end{equation}
\textcolor{black}{Applying~\Cref{lemma:Golowich}, to function $w$ at point $\hat{\vx}$ occurs}
\begin{equation}\label{eq:two}
\max_{\vy\in\calY} \sum_{i=1}^n U(\hat{\vx},\vy) - \min_{\vx\in\calX}\max_{\vy\in\calY}\sum_{i=1}^n U(\vx_i,\hat{\vx}_{-i},\vy) \leq \left(1+\frac{nL}{2\ell}\right)\norm{  \nabla w_{1/(2\ell)}(\hat{\vx}) }_2 \stackrel{(\ref{eq:smallw})}{=} O(\epsilon). 
\end{equation}
Further, we have
\begin{align}
    \min_{\vx\in\calX} \max_{\vy\in\calY} \sum_{i=1}^n U(\vx_i, \hat{\vx}_{-i},\vy) \leq  \max_{\vy\in\calY} \sum_{i=1}^n U(\vx^{\opt}_i,\hat{\vx}_{-i},\vy) &\leq \max_{\vy\in\calY} \sum_{i=1}^n U(\vx^{\opt}_i, \vx^*_{-i},\vy) + n L \frac{\epsilon}{2\ell} \label{align:Lip} \\ 
    &= nu^{\opt} + n L \frac{\epsilon}{2\ell},\label{eq:three}
\end{align}
where \eqref{align:Lip} follows by $(nL)$-Lipschitz continuity of $ \hat{\vx} \mapsto \max_{\vy\in \calY} \sum_{i=1}^n U(\vx_i^{\opt}, \hat{\vx}_{-i}, \vy)$, which in turn follows by $L$-Lipschitz continuity of $U(\vx, \vy)$ and \Cref{lemma:Lip}; and \eqref{eq:three} uses the fact that $n u^{\opt} =  \max_{b \in \calB} \sum_{i=1}^n U_b(\vx_i^{\opt}, \vx^*_{-i}) = \max_{\vy \in \calY} \sum_{i=1}^n U(\vx_i^{\opt}, \vx^*_{-i}, \vy)$, by optimality of $(u^{\opt}, \vx^{\opt})$. As a result, we conclude from \eqref{eq:star}, (\ref{eq:cont}), (\ref{eq:two}) and (\ref{eq:three}) that
\begin{equation}\label{eq:mainproof}
u^{\opt}\leq u^* \leq u^{\opt} + O(\epsilon). 
\end{equation}
Consider now the dual of the LP~\eqref{eq:primalLP} above, namely
\begin{equation}
    \label{eq:dual}
    \begin{aligned}
        & \mathrlap{\max ~\sum_{i=1}^n z_i} \\
        & \textrm{s.t.} & z_i \leq \sum_{b \in \calB} y(b) U_b(a, \vx^*_{-i}), & \textrm{ for all } i \in \calN, a \in \calA_i, \\
        & & \sum_{b\in\calB} y(b) =1, \\
        & & \vy\geq 0.  \\
    \end{aligned}
\end{equation}
Assume that $(\vz^{\opt},\vy^{\opt})$ is an optimal of the dual LP. Then, we observe that 
\begin{align}
    z^{\opt}_i = \sum_{a \in \calA_i} x^*_{i}(a) z^{\opt}_i &\leq \sum_{a \in \calA_i} \sum_{b\in\calB} x^*_{i}(a) y^{\opt}(b)  U_b(a,\vx^*_{-i}) \label{align-simfeas}
    \\ &=  \sum_{b\in\calB}  y^{\opt}(b)  \sum_{a \in \calA_i}  x^*_{i}(a)  U_b(a,\vx^*_{-i}) \nonumber 
    \\ &= \sum_{b\in\calB} y^{\opt}(b) U_b (\vx^*) \leq \sum_{b\in\calB} y^{\opt}(b) u^* = u^*, \label{align-last-feas}
\end{align}
where \eqref{align-simfeas} follows from the fact that $\vx_i^* \in \Delta(\calA_i)$ and by feasibility of $\vz^{\opt}$ in \eqref{eq:dual}; and \eqref{align-last-feas} follows since $u^* \geq U_b(\vx^*)$, by feasibility of $u^*$ in \eqref{eq:flp}, and the fact that $\vy^{\opt} \in \Delta(\calB)$, which in turn follows by feasibility of $\vy^{\opt}$ in the dual program \eqref{eq:dual}. Now, by strong duality, we have
\begin{equation}
    \label{eq:duality}
  \sum_{i \in \calN} z_i^{\opt} = n u^{\opt}.
\end{equation}
Combining with \eqref{eq:mainproof},
\begin{equation*}
    \sum_{i=1}^n z_i^{\opt} \geq n u^* - O(\epsilon),
\end{equation*}
in turn implying that
\begin{align}
    n U(\vx^*, \vy^{\opt}) = \sum_{i \in \calN} \sum_{b \in \calB} y^{\opt}(b) U_b(\vx^*) &= \sum_{i \in \calN} \sum_{a \in \calA_i} \sum_{b\in\calB} x^*_i(a) y^{\opt}(b) U_b(a, \vx_{-i}^*) \label{align:f-eq} \\
    &\geq \sum_{i \in \calN} \sum_{a \in \calA_i} x^*_i(a) z_i^{\opt} \geq \sum_{i \in \calN} z_i^{\opt} \geq n u^* - O(\epsilon), \label{align:f-sec}
\end{align}
where \eqref{align:f-eq} uses the linearity of $U(\vx, \vy)$ with respect to $\vy$ and multilinearity with respect to $\vx$; and \eqref{align:f-sec} follows from the fact that for any $i \in \calN$ and $a \in \calA_i$, it holds that $z^{\opt}_i \leq \sum_{b \in \calB} y^{\opt}(b) U_b(a, \vx^*_{-i})$, by feasibility of the pair $(\vz^{\opt}, \vy^{\opt})$ for the dual \eqref{eq:dual}. Given that $u^*$ is the utility the adversary obtains when best responding to $\vx^* \in \calX$, that is, $u^* = \max_{\vy \in \calY} U(\vx^*, \vy)$, \eqref{align:f-sec} establishes that $\vy^{\opt}$ is an $O(\epsilon)$-approximate best response to $\vx^*$.

Finally, let us bound the benefit of unilateral deviations from any team player. Strong duality~\eqref{eq:duality} along with \eqref{align-last-feas} yields that for any player $i \in \calN$,
\begin{equation*}
    z_i^{\opt} = n u^{\opt} - \sum_{i' \neq i} z_{i'}^{\opt} \geq n u^{\opt} - (n-1) u^* \geq u^* - O(\epsilon),
\end{equation*}
by \eqref{eq:mainproof}. But, by optimality of $(\vz^{\opt}, \vy^{\opt})$, it follows that 

$$z_i^{\opt} = \min_{a \in \calA_i} U(a, \vx^*_{-i}, \vy^{\opt}) = \min_{\vx_i \in \calX_i} U(\vx_i, \vx_{-i}^*, \vy^{\opt}).$$ 

We conclude that $(\vx^*,\vy^{\opt})$ is an $O(\epsilon)$-approximate Nash equilibrium. Rescaling $\epsilon$ with a sufficiently large $\poly(\Gamma)$ factor establishes~\eqref{eq:approxNE1}. Finally, the dual linear program~\eqref{eq:dual} has polynomially many constraints and variables, and its coefficients can be determined in polynomial time under the polynomial expectation property (\Cref{assumption:polexp}). This completes the proof.
\end{proof}

\subsection{The Potential Function and CLS Membership}
\label{sec:potential}

Leveraging \Cref{thm:onegeneralize}, here we establish the main result of this section: computing an $\epsilon$-approximate Nash equilibrium of an adversarial team game $\Gamma(n, U)$ belongs to the class $\CLS$, for any (polynomial) number of players $n \in \N$; $\CLS$-completeness is then guaranteed directly by the recent $\CLS$-hardness result of~\citet{babichenko2021settling} regarding identical-interest \emph{polytensor} (aka. hypergraphical) games, wherein the polynomial expectation property (\Cref{assumption:polexp}) clearly holds.

\begin{theorem}[\citealt{babichenko2021settling}]
    Computing an $\epsilon$-approximate Nash equilibrium in identical-interest games under \Cref{assumption:polexp} is \CLS-complete.
\end{theorem}

In this context, we first introduce $\gdmax$, described in detail below.

\paragraph{The algorithm.} The proposed algorithm, namely $\gdmax(\Gamma, \epsilon)$ (\Cref{alg:gdmax}), takes as input an adversarial team game $\Gamma$ and a precision parameter $\epsilon > 0$. The first step is to initialize all players' strategies at an arbitrary point $(\vx^{(0)}, \vy^{(0)}) \in \calX \times \calY$, and further set $T$, the number of iterations, to be sufficiently large $\poly(\Gamma)/\epsilon^4$. The algorithm then proceeds for $T$ iterations. In each step $t \in [T]$, we compute a best response for the adversary based on the current strategy for the team (\Cref{line:br}); under \Cref{assumption:polexp}, this step can be trivially performed in polynomial time by computing $\max_{b \in \calB} U(\vx^{(t-1)}, b)$. Next, based on the best response of the adversary, each player performs a projected gradient descent step (\Cref{line:gd}). More precisely, in \Cref{line:gd} we denote by $\Pi_{\calX_i}(\cdot)$ the Euclidean projection to the set $\calX_i$, which is well-defined since $\calX_i$ is nonempty, convex and compact. Also, since the strategy set $\calX_i$ of each player is a simplex, it is well-known that the projection can be computed exactly in nearly linear time in $|\calA_i|$. We further remark that the gradient for each player can be computed in polynomial time by virtue of \Cref{assumption:polexp}. Indeed, we observe that, by multilinearity, $\frac{\partial}{\partial x_i(a_i)} U(\vx, \vy) = \E_{\vec{a}_{-i} \sim \vec{x}_{-i}, b \sim \vec{y} } [U(a_i, \vec{a}_{-i}, b)]$. Now, based on the updated strategy of the team $\vx^{(t)}$, the response of the adversary is determined by solving the dual LP~\eqref{eq:dual} (\Cref{line:extend}) introduced earlier in the proof of \Cref{thm:onegeneralize}; more precisely, we replace $\vx^*$ in~\eqref{eq:dual} with the current strategy of the team $\vx^{(t)}$. In light of \Cref{thm:onegeneralize}, this is guaranteed to yield an $\epsilon$-approximate Nash equilibrium if $\vx^{(t)}$ is a sufficiently close approximate stationary point. This process is repeated until $(\vx^{(t)}, \vy^{(t)})$ is an $\epsilon$-approximate Nash equilibrium (\Cref{line:break}); in the sequel, we will show that $\poly(\Gamma)/\epsilon^4$ iterations suffice (\Cref{theorem:pot}) for the termination of the algorithm. Before we proceed, we summarize our problem below.

\begin{algorithm}[!ht]
  \caption{$\gdmax(\Gamma, \epsilon)$}
  \label{alg:gdmax}
  \begin{algorithmic}[1]
    \State \textbf{Input}: Adversarial team game $\Gamma$, precision parameter $\epsilon > 0$
    \State \textbf{Output}: An $\epsilon$-approximate Nash equilibrium of $\Gamma$
    \State Initialize $(\vec{x}^{(0)},\vec{y}^{(0)}) \in \calX \times \calY$ and $T = \poly(\Gamma)/\epsilon^4$ \label{line:init}
    \For{$t \gets 1,2,\dots, T$}
    \If{$(\vec{x}^{(t-1)},\vec{y}^{(t-1)})$ is an $\epsilon$-approximate NE} \label{line:break}
      \State \textbf{break}
      \EndIf
      \State {$\ay^{(t)}\gets  \arg \max_{\vec{y}\in \calY} U (\vec{x}^{(t-1)},\vec{y})$} \label{line:br}
      \State {$\vec{x}^{(t)}_{i} \gets
                            \Pi_{\calX_i}\left\{ 
                                \vec{x}_{i}^{(t-1)}-\eta  \nabla_{\vec{x}_i} U \big( {\vec{x}^{(t-1)},\ay^{(t)}}
                            \big)
                                \right\} \label{line:gd}
                        $} \Comment{for all players $i \in \calN$ }
      \State $\vec{y}^{(t)} \gets \texttt{ExtendNE}(\vec{x}^{(t)})$ \label{line:extend} \Comment{Solve dual LP$(\vx^{(t)})$~\eqref{eq:dual} }
      \EndFor\\
      \Return{$(\vec{x}^{(t)},\vec{y}^{(t)})$}
  \end{algorithmic}
\end{algorithm}

\begin{nproblem}[\textsc{TeamVsAdversary-Nash}]
  \textsc{Input:} A polynomial-time Turing machine (TM) $\calC_U$ for the succinct representation of the utility (see \Cref{assumption:polexp}) and an approximation parameter $\epsilon > 0$.
  \smallskip

  \noindent \textsc{Output:} A strategy profile $(\vec{x}^*,\vec{y}^*) \in \calX \times \calY$ which is an $\epsilon$-approximate NE for the game $\Gamma(n, U)$.
\end{nproblem}

\textcolor{black}{Above, $\log (1/\epsilon)$ is polynomial in the input, so that the approximation parameter $\epsilon > 0$ can be represented with a polynomial number of bits.}

For the sake of self-inclusiveness, we define the necessary complexity classes $\FNP$ and $\CLS$, as well as the notion of reductions that we use in this paper to show the membership or hardness of a problem for one of these complexity classes. To formally define a search problem, we use the standard notation $\{0, 1\}^*$ to denote the set of all finite length bit-strings.

\begin{definition}[Search Problems - ${\FNP}$] \label{def:FNP}
    A binary relation $\calR \subseteq \{0, 1\}^* \times \{0, 1\}^*$ is in
  the class $\FNP$ if (i) for every $\vec{p}, \vec{q} \in \{0, 1\}^*$ such that
  $(\vec{p}, \vec{q}) \in \calR$, it holds that 
  $\abs{\vec{q}} \le \poly(\abs{\vec{p}})$; and (ii) there exists an algorithm that
  verifies whether $(\vec{p}, \vec{q}) \in \calR$ in time 
  $\poly(\abs{\vec{p}}, \abs{\vec{q}})$. The \textit{search problem} associated with
  a binary relation $\calR$ takes some $\vec{p}$ as input and requests as output
  some $\vec{q}$ such that $(\vec{p}, \vec{q}) \in \calR$, or outputs $\bot$ if no
  such $\vec{q}$ exists. 
  The \textit{decision problem} associated with $\calR$
  takes some $\vec{p}$ as input and requests as output the bit $1$, if there
  exists some $\vec{q}$ such that $(\vec{p}, \vec{q}) \in \calR$, and the bit $0$
  otherwise. The class $\NP$ is defined as the set of decision problems
  associated with relations $\calR \in \FNP$.
\end{definition}

\begin{definition}[Polynomial-Time Reductions]
    A search problem $P_1$ is \textit{polynomial-time reducible} to a search
  problem $P_2$ if there exist polynomial-time computable functions 
  $f : \{0, 1\}^* \to \{0, 1\}^*$ and 
  $g : \{0, 1\}^* \times \{0, 1\}^* \times \{0, 1\}^* \to \{0, 1\}^*$
  with the following properties: (i) if $\vec{p}$ is an input to $P_1$, then
  $f(\vec{p})$ is an input to $P_2$; and (ii) if $\vec{q}$ is a solution to $P_2$ on
  input $f(\vec{p})$, then $g(\vec{p}, f(\vec{p}), \vec{q})$ is a solution to $P_1$ on
  input $\vec{p}$.
\end{definition}

Finally, leveraging the recent results of \citet{Fearnley21:The}, $\CLS$ can be simply defined as the set of all problems in $\FNP$ that belong to both $\PLS$ and $\PPAD$. Thus,  we start our $\CLS$ inclusion by recalling the seminal work of~\citet{papadimitriou1994complexity,Daskalakis06:The}, which established that the problem of computing an $\epsilon$-approximate Nash equilibrium of any $n$-player normal-form game belongs to $\PPAD$.

\begin{corollary}\label{cor:PPAD-inclusion}
\textsc{TeamVsAdversary-Nash} is in $\PPAD$.
\end{corollary}

Hence, our proof focuses on showing that \textsc{TeamVsAdversary-Nash} belongs to $\PLS$, or equivalently, that it is polynomial-time reducible to a modification of \textsc{General-Real-LocalOpt}~\citep{Fearnley21:The,Daskalakis11:Continuous}, introduced below using polynomial-time TMs.

\begin{nproblem}[\kappa-\textsc{General-Real-LocalOpt-TM}]
  \textsc{Input:} \begin{itemize}
  \item Precision/stopping parameter $\epsilon'>0.$
  \item A bounded and nonempty domain $D = [0, 1]^k$, for $k \in \N$.
 \item Two polynomial-time Turing machines (TMs) $\mathcal{C}_h$ and $\mathcal{C}_w$ which evaluate the functions $h : \mathbb{R}^k \to \mathbb{R}$ and $w : \mathbb{R}^k \to \mathbb{R}^k$, constrained to have running time bounded by $z^\kappa$, where $z$ is the input size of the TMs and $\kappa \in \N$.  
  \item It is also promised that the function $h$ is Lipschitz continuous.
  \end{itemize}
  \smallskip

  \noindent \textsc{Output:} Compute an approximate local optimum of $h$ with respect to $w$ on domain $D.$ Formally, find $\vec{x}\in D$ such that  $h\left(\textcolor{black}{\Pi_{D}}\left(w(\vec{x})\right)\right) \geq h(\vec{x}) - \epsilon'.$
\end{nproblem}

\begin{remark}
    \label{remark:TM}
Let us briefly point out a subtle technical issue in the definition above. To study the computational complexity of a problem that takes as input a polynomial-time Turning machine, one can syntactically enforce an upper bound on the computation; this is precisely the role of parameter $\kappa \in \N$ in the definition above. In particular, we will say that $\textsc{General-Real-LocalOpt-TM}$ belongs to a complexity class if for any fixed $\kappa \in \N$, $\kappa-\textsc{General-Real-LocalOpt-TM}$ belongs to that complexity class. This issue is also discussed by~\citet[Remark 2.6]{Daskalakis21:The}. 
\end{remark}

We now claim that this problem lies in the class $\PLS$:

\begin{lemma}\label{lem:probpls} $\textsc{General-Real-LocalOpt-TM}$ is in $\PLS$.
\end{lemma}

\begin{proof}
The argument proceeds by discretizing the domain as in the proof of~\citet[Theorem 2.1]{Daskalakis11:Continuous}, which directly applies for the domain $[0, 1]^k$ as noted by~\citet[Proposition D.2]{Fearnley21:The}.
\end{proof}

To proceed with our reduction, we first define (with a slight abuse of notation) the notion of approximate proximal point that will be used in the definition of the potential function.

\begin{definition}[Approximate Proximal Point]
    \label{def:app}
We denote by $\prox_{\phi_{1/(2\ell)}}(\vec{x}; \epsilon)$ any $\epsilon$-approximate solution in value to the program $$\min_{\vec{x}' \in \calX} \phi(\vec{x}') + \ell\norm{\vec{x}-\vec{x}'}_2^2,$$ where $\phi(\vec{x}):=\max_{\vec{y} \in \calY} U(\vec{x},\vec{y}).$
\end{definition}
\noindent Below, we prove that we can indeed compute an $\epsilon$-approximate solution in polynomial time.
\newcommand{\algo}{\mathcal{A}lg}
\begin{theorem}\label{th:PLS-part1}
Let function $V_{\vec{x}}:\calX\to\mathbb{R}$ be defined as $V_{\vec{x}}(\vec{x}')=\max_{\vec{y} \in \calY} U(\vec{x}',\vec{y}) + \ell\norm{\vec{x}-\vec{x}'}_2^2$. Then, there exists a $\poly(n, \sum_{i=1}^n |\calA_i|, |\calB|,\log(1/\epsilon))$-time algorithm $\algo: \calX\to\calX$ such that 
$$\min_{\vec{x}'\in\calX} V_{\vec{x}}(\vec{x}')\leq V_{\vec{x}}(\algo(\vec{x})) \leq \min_{\vec{x}'\in\calX} V_{\vec{x}}(\vec{x}')+\epsilon.$$ 
\end{theorem}

\begin{proof}
First, we will show that a subgradient of $V_{\vec{x}}(\vec{x}')$ for any rational input $\vec{x}'$ is rational and can be computed both efficiently and accurately.
For the first summand, we recall Danskin's theorem to compute the gradient of $\max_{\vec{y} \in \calY} U(\vec{x}',\vec{y})$.

\begin{fact}[Generalized Danskin Theorem, \textcolor{black}{ \cite[Theorem 3]{Daskalakis20:Independent}}]
    \label{fact:Danskin}
\[\partial \phi(\vec{x}') = \partial \max_{\vec{y} \in \calY} U(\vec{x}',\vec{y})=\textcolor{black}{\textrm{conv}}\left\{ \nabla_{\vec{x}} U(\vec{x}',\vec{y}^\star(\vec{x}')) \Big|\ \  \vec{y}^\star \in \arg\max_{\vec{y} \in \calY} U(\vec{x}',\vec{y}) \right\}.\]
\end{fact}

For a fixed $\vec{x}'$, we can trivially compute a rational best response strategy $\vec{y}^{\star}(\vec{x}')\in \calY$ for the adversary (using \Cref{assumption:polexp}). Hence, using the oracle access to the utility function, we have access to a subgradient for $\phi(\vec{x}')$. For the second summand, we have that $\partial_{\vec{x}} \{\ell \cdot(\|\vec{x}-\vec{x}'\|^2_2)\}=2\ell(\vec{x}-\vec{x'})$. 
Therefore, the latter argument implies that for a given $\vec{x}'$, we can compute efficiently a subgradient of $V_{\vec{x}}(\cdot)$. 
Additionally, for the constraint set $\calX$ we can have easily a strong separation oracle as it is a product of simplices. Therefore, either via subgradient cuts of central-cut ellipsoid method~\citep[Chapter 2 \& 3]{grotschel2012geometric}, or more modern cutting plane methods~\citep{lee2015faster}, we can compute in polynomial time a point $\hat{\vec{x}}\in \calX$ with rational coordinates such that 
$\min_{\vec{x}'\in\calX} V_{\vec{x}}(\vec{x}')\leq V_{\vec{x}}(\hat{\vec{x}}) \leq \min_{\vec{x}'\in\calX} V_{\vec{x}}(\vec{x}')+\epsilon$.
\end{proof}

Equipped with the above theorem, the following proposition presents a polynomially computable potential function for the iterations of \Cref{alg:gdmax}:

\begin{theorem}
    \label{theorem:pot}
    Consider any precision $\epsilon > 0$, and let 
    $$g : \calX \ni \vec{x} \mapsto \max_{\vec{y}' \in \calY} U(\prox_{\phi_{1/(2\ell)}}(\vec{x}; O(\epsilon^4)), \vec{y}') + \ell\norm{\vec{x} - \prox_{\phi_{1/(2\ell)}}(\vec{x}; O(\epsilon^4))}_2^2.$$ For any adversarial team game $\Gamma(n, U)$, $\gdmax(\Gamma, \epsilon)$ with a sufficiently small learning rate $\eta = \Theta(\epsilon^2)$ satisfies
    $$g(\vec{x}^{(t)}) < g(\vec{x}^{(t-1)}) - \Omega(\epsilon^4),$$ 
    for any iteration $t$ until termination.
\end{theorem}

In the following argument, we utilize techniques by~\citet{lin2020gradient} in order to argue that the team players reach an approximate stationary point (in the sense of \Cref{def:stat}).

\begin{proof}[Proof of \Cref{theorem:pot}]
Let us fix a number of iterations $T \in \N$. By $\ell$-smoothness of $U(\vec{x},\vec{y})$, we have that for all $\vec{x} \in \calX$ and $t \in [T]$,
\begin{equation}\label{eq:smooth1}
    \max_{\vec{y}' \in \calY} U(\vec{x},\vec{y}') \geq U(\vec{x}, \ay^{(t)}) \geq U(\vec{x}^{(t)},\ay^{(t)}) +\langle\nabla_{\vx} U(\vec{x}^{(t)},\ay^{(t)}),\vec{x}-\vec{x}^{(t)}\rangle-\frac{\ell}{2}\norm{\vec{x}-\vec{x}^{(t)}}_2^2.
\end{equation}
We recall that $\phi : \vx \mapsto \max_{\vy \in \calY} U(\vx, \vy)$. For any iteration $t \in [T-1]$,
\begin{align}
        g(\vec{x}^{(t+1)}) &=\max_{\vec{y} \in \calY} U\left( \prox_{\phi_{1/(2\ell)}}(\vec{x}^{(t+1)}; O(\epsilon^4)), \vec{y}\right) + \ell \norm{\prox_{\phi_{1/(2\ell)}}(\vec{x}^{(t+1)}; O(\epsilon^4))-\vec{x}^{(t+1)}}_2^2 \label{align:pot-def} 
        \\ &\leq \max_{\vec{y} \in \calY}  U(\tilde{\vec{x}}^{(t)},\vec{y}) + \ell \norm{\tilde{\vec{x}}^{(t)}-\vec{x}^{(t+1)}}_2^2 + O(\epsilon^4) \label{align:approx-val}
        \\
            &\leq \max_{\vec{y} \in \calY} U(\tilde{\vec{x}}^{(t)},\vec{y}) + \ell \norm{ \tilde{\vec{x}}^{(t)} - \left( \vec{x}^{(t)}-\eta \nabla_{\vec{x}}U(\vec{x}^{(t)},\ay^{(t)})\right) }_2^2 + O(\epsilon^4) \label{align:nonexp}
        \\
        &\leq g (\vec{x}^{(t)}) + O(\epsilon^4) + \eta^2 \ell L^2\nonumber
        \\&+ 2\eta \ell \left(\max_{\vec{y}\in\calY} U(\tilde{\vec{x}}^{(t)},\vec{y})-\max_{\vec{y} \in \calY} U(\vec{x}^{(t)},\vec{y})+\frac{\ell}{2}\norm{\vec{x}^{(t)}-\tilde{\vec{x}}^{(t)}}_2^2\right), \label{align:last}
\end{align}
where \eqref{align:pot-def} follows from the definition of $g$; \eqref{align:approx-val} uses the definition of $\prox_{\phi_{1/(2\ell)}}(\vec{x}^{(t)}; O(\epsilon^4))$ (\Cref{def:app}), with the convention that $\tilde{\vx}^{(t)} \defeq \prox_{\phi_{1/(2\ell)}}(\vx^{(t)})$; \eqref{align:nonexp} uses the definition of $\vx^{(t+1)}$ (\Cref{line:gd} in \Cref{alg:gdmax}), along with the fact that the (Euclidean) projection operator $\Pi(\cdot)$ is nonexpansive with respect to $\|\cdot\|_2$; and \eqref{align:last} uses \eqref{eq:smooth1} for $\vx \defeq \tilde{\vx}^{(t)}$, the property that $\| \nabla_{\vx} U(\vx^{(t)}, \ay^{(t)})\|_2 \leq L$ (by $L$-Lipschitz continuity of $U$), and the fact that $\max_{\vy' \in \calY} U(\tilde{\vx}^{(t)}, \vy') + \ell \| \tilde{\vx}^{(t)}-  \vx^{(t)}\|_2^2 \leq g(\vx^{(t)})$ (by definition of $\tilde{\vx}^{(t)}$). Further, since $\max_{\vec{y}' \in \calY} U(\vec{x},\vec{y}') + \ell \norm{\vec{x} - \vec{x}^{(t)}}_2^2$ is $\ell$-strongly convex (by \Cref{lemma:weaklyconvex}), we have
\begin{gather}
\max_{\vec{y} \in \calY} U(\vec{x}^{(t)},\vec{y}) - \max_{\vec{y} \in \calY} U(\tilde{\vec{x}}^{(t)},\vec{y}) - \frac{\ell}{2} \norm{\vec{x}^{(t)}-\tilde{\vec{x}}^{(t)}}_2^2 \geq \ell \norm{\vec{x}^{(t)}-\tilde{\vec{x}}^{(t)}}_2^2,
\end{gather} 
since $\tilde{\vx}^{(t)} = \arg \min_{\vx \in \calX} \max_{\vy \in \calY} U(\vx, \vy) + \ell \norm{\vx - \vx^{(t)}}_2^2$. Thus, combining with \eqref{align:last}, we conclude that
\begin{equation}
    \label{eq:pot}
g(\vec{x}^{(t+1)}) \leq g(\vec{x}^{(t)}) - 2\eta \ell^2 \norm{\vec{x}^{(t)}-\tilde{\vec{x}}^{(t)}}_2^2+ \eta^2\ell L + O(\epsilon^4).
\end{equation}

Now for any iteration $t+1$ for which the algorithm does not terminate, we know from \Cref{thm:onegeneralize} that $\norm{\vec{x}^{(t)}-\tilde{\vec{x}}^{(t)}}_2 = \Omega(\epsilon)$. Indeed, if $\norm{\vec{x}^{(t)}-\tilde{\vec{x}}^{(t)}}_2 \leq \epsilon/\poly(\Gamma)$, for a sufficiently large $\poly(\Gamma)$, \Cref{thm:onegeneralize} along with \Cref{lemma:close-prox} would imply that $(\vx^{(t)}, \vy^{(t)})$ is an $\epsilon$-approximate Nash equilibrium, which would in turn terminate the algorithm. As a result, for a sufficiently small learning rate $\eta = O(\epsilon^2)$, it follows from \eqref{eq:pot} that
\begin{equation}
    \label{eq:ratepot}
    g(\vx^{(t+1)}) \leq g(\vx^{(t)}) - \Omega(\epsilon^4) < g(\vx^{(t)}).
\end{equation}

This shows that $g$ is a proper potential function. Further, given that $g$ is bounded by $\poly(\Gamma)$, we conclude that after a sufficiently large number of iterations $T = \poly(\Gamma, 1/\epsilon)$ the algorithm will terminate. This concludes the proof.
\end{proof}

Given that the potential function $g$ is bounded by $\poly(\Gamma)$ and the decrease per iteration is $\Omega(\epsilon^4)$~\eqref{eq:ratepot}, we also obtain as a byproduct a (pseudo) fully polynomial-time approximation scheme (\FPTAS) for computing an approximate Nash equilibrium in adversarial team games:

\begin{corollary}[Pseudo-$\FPTAS$ for Approximate Nash]
    \label{cor:fptas}
    Consider any precision $\epsilon > 0$. For any adversarial team game $\Gamma(n, U)$, $\gdmax(\Gamma, \epsilon)$ with a sufficiently small learning rate $\eta = O(\epsilon^2)$ yields an $\epsilon$-approximate Nash equilibrium after a sufficiently large $T = \poly(\Gamma)/\epsilon^4$. Further, under \Cref{assumption:polexp}, every iteration of $\gdmax$ can be implemented in polynomial time. 
\end{corollary}

In particular, we note that under the assumption that the maximum absolute value of a utility in the payoff tensor is bounded by $\poly(\sum_{i=1}^n |\calA_i|, |\calB|)$, the corollary above yields an $\FPTAS$. We next proceed with the proof of our main result.

\begin{theorem} 
    \label{theorem:cls}
The problem of computing an $\epsilon$-approximate Nash equilibrium in adversarial team games under \Cref{assumption:polexp} is $\CLS$-complete.
\end{theorem}

\begin{proof}
We want to show that, for some constant $\kappa \in \N$,
\[
\textsc{TeamVsAdversary-Nash} \leq_{P} \kappa-\textsc{General-Real-LocalOpt-TM}.
\] 
Given an instance of an adversarial team game $\Gamma(n, U)$ with precision parameter $\epsilon > 0$, we construct an instance of a $\textsc{General-Real-LocalOpt-TM}$ as follows. We set $D \defeq [0, 1]^k$, where $k \defeq \sum_{i=1}^n |\calA_i| + |\calB|$. Further, we let $h(\vx, \vy) \defeq h(\vx) \defeq \min_{\vx' \in \calX} \left\{ \phi(\vx') + \ell \|\vx - \vx'\|_2^2 \right\}$ be the potential function, which is Lipschitz continuous and can be evaluated with at most $\epsilon$-error in time $\poly(n, \sum_{i=1}^n |\calA_i|, |\calB|, \log(1/\epsilon))$ (\Cref{th:PLS-part1}). Finally, the successor $w(\vec{x},\vec{y})$ is defined in \Cref{line:gd,line:extend} of $\gdmax(\Gamma, \epsilon)$, which we have shown can be computed in time $\textrm{poly}(n, \sum_{i=1}^n |\mathcal{A}_i|, |\calB|)$ under \Cref{assumption:polexp}. Namely, if $\vy' \defeq \arg\max_{\vy \in \calY} U(\vx, \vy)$ \textcolor{black}{and $\vx' \defeq \Pi_{\calX} \left\{ \vx - \eta \nabla_{\vx} U(\vx, \vy') \right\}$,}
\begin{equation*}
    w(\vx, \vy) \defeq \left( \Pi_{\calX} \left\{ \vx - \eta \nabla_{\vx} U(\vx, \vy') \right\}, \texttt{ExtendNE}(\textcolor{black}{\vec{x}'}) \right) \in [0, 1]^k,
\end{equation*}
for a sufficiently small $\eta = O(\epsilon^2)$, where we recall that $\texttt{ExtendNE}(\vec{x})$ gives the solution of the polynomial LP~\eqref{eq:dual}. We note that both $h$ and $w$ can be implemented via polynomial-time TMs $\mathcal{C}_h$ and $\mathcal{C}_w$ (polynomial in the input size and $\log (1/\epsilon)$). In particular, for a suitable constant value $\kappa \in \N$, their running time will be bounded by $z^{\kappa}$, where $z$ is their input size. 

Now if $h(w(\vx, \vy)) \geq h(\vx, \vy) - \epsilon'$, for a sufficiently small $\epsilon' = \epsilon^4 \times \poly(\Gamma)$, \Cref{theorem:pot} implies that $(\vx, \vy)$ is an $\epsilon$-approximate Nash equilibrium. Thus, by \Cref{lem:probpls}, it follows that \textsc{TeamVsAdversary-Nash} is in $\PLS$.




Combining this with Corollary~\ref{cor:PPAD-inclusion}, we conclude that 
\textsc{TeamVsAdversary-Nash} belongs to $\PPAD \cap \PLS$. Thus, leveraging the recent result of~\citet{Fearnley21:The}, it holds that 
\textsc{TeamVsAdversary-Nash} belong to $\CLS$. Finally, $\CLS$-hardness follows directly by the hardness result of~\citet{babichenko2021settling}.
\end{proof}

\begin{remark}[Correlated Adversaries]
Another notable application of having a single adversary is the case where the adversary team has multiple players, but with the twist that the adversaries are allowed to correlate their strategies---\emph{i.e.}, the team is facing a ``virtual'' player. However, in that case the action space of that virtual player grows exponentially with the number of adversaries $m$, and so establishing polynomial-time algorithms when $m \gg 1$ requires further work.
\end{remark}

\begin{remark}[\textcolor{black}{On the Algorithm and the potential}] \textcolor{black}{In this section we showed that if the team players perform a step of gradient descent on the best response of the adversary, the Moreau envelop $\phi(\vx)$ of $\max_{\vy\in\calY} U(\vx,\vy)$ is strictly decreasing (Theorem \ref{theorem:pot}), for appropriate step-sizes. We note that if one instead performs the simple proximal point method, i.e.,}
$$\textcolor{black}{\vec{x}^{(t+1)} = \arg \min_{\vx\in \calX}\left\{\max_{\vy\in\calY}U(\vx,\vy) +\frac{1}{2\lambda}\norm{\vx-\vx^{(t)}}\right\}},$$
\textcolor{black}{then the Moreau envelop remains a potential function. It is noteworthy that the proximal point method can be interpreted as gradient descent applied to the Moreau envelope.}
\end{remark}

\section{Additional Applications of Two-Team Zero-Sum Games}
\label{sec:advcong}

In this section, we discuss about some additional interesting settings that can be cast as two-team zero-sum games. For the sake of generality, here we allow the adversary team to have more than a single player.

\paragraph{Adversarial potential games.} Consider a normal form game with two sets of players $\calN$ and $\calM$. Following the conventions of \Cref{sec:teams}, the minimizers $\calN$ will be associated with cost functions, while the maximizers $\calM$ will be associated with utilities. Further, there exists an underlying \emph{potential function} $\Phi : \calX \times \calY \to \R$, so that 
\begin{itemize}[noitemsep]
    \item[(i)] for each minimizer $i \in \calN$ and strategies $\vec{x}_i,\vec{x}_i' \in \calX_i$, 
    \begin{equation}
        \label{eq:cost-pot}
    c_i(\vec{x}_{i}',\vec{x}_{-i},\vec{y}) - c_i(\vec{x}_{i},\vec{x}_{-i},\vec{y}) = \Phi(\vec{x}_{i}',\vec{x}_{-i},\vec{y}) - \Phi(\vec{x}_{i},\vec{x}_{-i},\vec{y});    
    \end{equation}
    \item[(ii)] for each maximizer $j \in \calM$ and strategies $\vec{y}_j,\vec{y}_j' \in \calY_j$,
    \begin{equation}
        \label{eq:util-pot}
        u_j(\vec{x},\vec{y}_{j}',\vec{y}_{-j}) - u_j(\vec{x},\vec{y}_{j},\vec{y}_{-j}) = \Phi(\vec{x},\vec{y}_{j}',\vec{y}_{-j}) - \Phi(\vec{x},\vec{y}_{j},\vec{y}_{-j}).
    \end{equation}
\end{itemize}
While this setting is not a two-team zero-sum game, it can still be captured by one, namely $\Gamma(n,m,\Phi)$. More precisely, \eqref{eq:cost-pot} implies that $c_i (\vec{x},\vec{y}) = \Phi(\vec{x},\vec{y}) + \Psi_i(\vec{x}_{-i},\vec{y})$, for any $i \in \calN$ and $(\vx, \vy) \in \calX \times \calY$, where $\Psi_i$ does not depend on $\vec{x}_i$, and \eqref{eq:util-pot} implies that $u_j (\vec{x},\vec{y}) = \Phi(\vec{x},\vec{y}) + \Psi'_j(\vec{x},\vec{y}_{-j})$, where $\Psi_j$ does not depend on $\vec{y}_j$. As a result, we arrive at the following conclusion.

\begin{proposition}
    A strategy profile satisfying \eqref{def:vi} with respect to the two-team zero-sum game $\Gamma(n, m, \Phi)$ will be an $\epsilon$-approximate Nash equilibrium of the original adversarial potential game.
\end{proposition}

\paragraph{Congestion games with adversarial costs.}

A \emph{congestion game}, dating back at least to the work of~\citet{rosenthal73}, is defined by the tuple $\Gamma = (\mathcal{N}; \calE;$ $(\calA_i)_{i \in \mathcal{N}};(c_e)_{e \in \calE})$, where
\begin{itemize}[noitemsep]
    \item $\mathcal{N}$ is the set of \emph{agents} with $n = |\mathcal{N}|$;
    \item $\calE$ is a set of \emph{edges} (also known as \emph{resources} or \emph{bins} or \emph{facilities});
    \item the action space $\calA_i$ of each player $i \in \calN$ is a collection of subsets of $\calE$ ($\calA_i \subseteq 2^{\calE}$); and
    \item $c_e : \N \cup \{0\} \to \R_{> 0}$ is the cost (\emph{e.g.}, latency) function associated with edge $e \in \calE$.
\end{itemize}
More precisely, for an action profile $\vec{a} = (a_1, \dots, a_n) \in \prod_{i=1}^n \calA_i$, the cost of player $i \in \calN$ is given by $c_i(\vec{a}) \defeq \sum_{e \in \calA_i} c_e(\ell_e(\vec{a}))$, where $\ell_e(\vec{a})$ denotes the number of players using $e \in \calE$ under $\vec{a}$---the load of edge $e$.

Now we introduce a twist to this well-studied setting: there is an additional ``adversarial'' player who chooses the cost functions of the edges from a finite family of functions $\mathcal{F}_e = \{c_{e,1},...,c_{e,m_e}\}$, with $|\mathcal{F}_e| = m_e$. We denote by $\vec{b}$ the action of the adversary, so that $b_e \in \calF_e$ denotes the cost function $c_{e,b_e}$ the adversary chose for edge $e \in \calE$. In this context, every agent is aiming at minimizing its cost, while the adversary is aiming at maximizing the potential function $\Phi(\vec{a},\vec{b}) \defeq \sum_{e \in \calE} \sum_{j=0}^{\ell_e(\vec{a})} c_{e,b_e}(j)$. A similar setting is the case of nonatomic congestion games with malicious players presented by \citet{Babaioff09:Congestion}; apart from the players that strive to minimize their cost, part of the network's flow is controlled by a malicious (adversarial) player that attempts to maximize the load experienced by the set of the cost-minimizing players.

This setting can be equivalently formulated as follows: there is a collection of congestion games (with respect to the team players), and the adversary has the power of selecting the underlying congestion game. As such, computing Nash equilibria is easily seen to be cast in the framework of \Cref{sec:cls}. One caveat here is that the complexity of our algorithms depends polynomially on the number of actions, which could be exponential in the description of the congestion game. One promising approach to address this issue is to employ \emph{multiplicative weights} instead of gradient-descent-type dynamics, as in \Cref{alg:gdmax}, in conjunction with the \emph{kernel trick}~\citep{Takimoto03:Takimoto} when the underlying action set is structured~\citep{Beaglehole23:Sampling}.

\section{Conclusions and Open Problems}
\label{sec:conclusion}

Our main contribution in this paper was to establish that computing an approximate Nash equilibrium in adversarial team games is $\CLS$-complete. This further expands the class of problems known to be complete for that class, and extends the recent $\CLS$-completeness of~\citet{babichenko2021settling} (for potential games) to an important class of games that combines both cooperation and competition; we argue that understanding the complexity of such settings serves as an important step to demystify more realistic multi-agent environments. 

While we have focused on the setting where the team is facing a single adversary, the canonical case treated in the literature starting from the work of~\citet{VonStengel97:Team} (recall our overview in~\Cref{sec:related}), a natural question is to investigate the complexity of computing Nash equilibria in the more general setting where both teams have multiple players. In this context, we make the following conjecture.

\begin{conjecture}
    \label{conjecture:PPAD}
    Computing an $\epsilon$-approximate Nash equilibrium, for a sufficiently small $\epsilon > 0$, in two-team zero-sum games with $n, m \gg 1$ is $\PPAD$-complete.
\end{conjecture}

Even the complexity status of the problem when $n, m = 2$ remains a challenging open problem. An important implication of any hardness result for that problem would be on the computation of equilibria in nonconvex-nonconcave min-max optimization~\citep{Daskalakis21:The}, given that computing equilibria in two-team zero-sum games can be viewed as a special case of the former problem. In particular, the hardness of~\citet{Daskalakis21:The} rests on the assumption that the players' strategy sets are coupled, while answering~\Cref{conjecture:PPAD} necessitates going from coupled to product constraints, which is recognized as a major challenge. An even more basic question: Is there a $\PTAS$ for computing Nash equilibria when both teams have multiple players? 

Returning to the canonical case of a single adversary, it would be interesting to understand the complexity of computing (approximate) Nash equilibria when $n = O(1)$. While our $\CLS$ membership clearly applies in that case, one can no longer rely on the work of~\citet{babichenko2021settling} to obtain a non-trivial hardness result given that Nash equilibria in identical-interest games with a constant number of players can be trivially computed in polynomial time (by simply identifying the maximum entry in the payoff tensor). In fact, even the case where $n = 2$ is a challenging open problem.

Moreover, the situation becomes even bleaker when insisting on an exact equilibrium, since one has to deal with the inherent irrationality of equilibria in adversarial team games~\citep{VonStengel97:Team}. For this reason, $\FIXP$, introduced by~\citet{Etessami10:On}, seems to be a plausible candidate for capturing the complexity of the problem:

\begin{question}
    \label{conj:FIXP}
    Is computing an exact Nash equilibrium in adversarial team games $\FIXP$-complete?
\end{question}

\textcolor{black}{ There is some negative---albeit somewhat superficial---evidence against \FIXP-completeness: when $\FIXP$ is viewed as the counterpart of $\PPAD$ for the exact version of the underlying problem, then one could also conjecture that computing an exact Nash equilibrium in adversarial team games is instead complete for a more benign class than \FIXP---one in correspondence to \CLS. In fact, a simpler but yet long-standing open question in this area is to understand the complexity of computing an exact (mixed) Nash equilibria in potential games.}

Finally, our focus in this paper has been on games represented (succinctly) in normal form. It would be interesting to extend our $\CLS$ membership to other classes of games which are not of \emph{polynomial type}, in that the number of actions of each player is superpolynomial, most notably to extensive-form games. One major challenge that arises there is that the per-iteration complexity of gradient-based dynamics is no longer polynomial, thereby necessitating a more refined approach that leverages the underlying structure.

\section*{Acknowledgements}

We are grateful to the anonymous reviewers at the conference on Economics and Computation (EC) 2023 for their valuable feedback, and in particular for spotting a technical inconsistency in an earlier version of this manuscript. We also thank Alexandros Hollender for insightful discussions regarding \Cref{remark:TM}; all errors remain our own.

\bibliography{main}

\appendix
\section{Omitted Proofs}
\label{appendix:proofs}

In this section, we provide the omitted proofs of some simple claims we made in the main body. For the convenience of the reader, we will restate each claim before proceeding with its proof. We commence with \Cref{lemma:Lipcon}.

\begin{lemma}
    \label{lemma:Lipcon}
    Suppose that $|U(\vec{a}, \vec{b})| \leq V$ for any $(\vec{a}, \vec{b}) \in \calA \times \calB$. Then, for any $(\vx, \vy), (\vx', \vy') \in \calX \times \calY$, 
    \begin{equation*}
        | U(\vx, \vy) - U(\vx', \vy')| \leq V \sqrt{ \sum_{i=1}^n |\calA_i| + |\calB|} \|(\vx, \vy) - (\vx', \vy')\|_2.
    \end{equation*}
\end{lemma}

\begin{proof}
By definition, we have
    \begin{align}
        | U(\vx, \vy) - U(\vx', \vy')| &= \left| \E_{(\vec{a}, \vec{b}) \sim (\vx, \vy)} U(\vec{a}, \vec{b}) - \E_{(\vec{a}, \vec{b}) \sim (\vx', \vy')} U(\vec{a}, \vec{b}) \right| \notag \\
        &= \left| \sum_{(\vec{a}, \vec{b}) \in \calA \times \calB } U(\vec{a}, \vec{b}) \left( \prod_{i \in \calN} x_i(a_i) \prod_{j \in \calM} y_j(b_j) -  \prod_{i \in \calN} x'_i(a_i) \prod_{j \in \calM} y'_j(b_j) \right) \right| \notag \\
        &\leq V \left| \sum_{(\vec{a}, \vec{b}) \in \calA \times \calB } \left( \prod_{i \in \calN} x_i(a_i) \prod_{j \in \calM} y_j(b_j) -  \prod_{i \in \calN} x'_i(a_i) \prod_{j \in \calM} y'_j(b_j) \right) \right| \label{eq:tri} \\
        &\leq V \left( \sum_{i \in \calN} \|\vx_i - \vx_i'\|_1 + \sum_{j \in \calM} \|\vy_j - \vy_j'\|_1 \right) \label{eq:tv} \\
        &\leq V \sqrt{ \sum_{i=1}^n |\calA_i| + |\calB|} \|(\vx, \vy) - (\vx', \vy')\|_2, \label{eq:equiv-norm}
    \end{align}
    where \eqref{eq:tri} follows by the triangle inequality and the assumption that $|U(\vec{a}, \vec{b})| \leq V$; \eqref{eq:tv} uses the property that the total variation distance between two product distributions is the sum of the total variation distance of the marginals~\citep{Hoeffding58:Dis}; and \eqref{eq:equiv-norm} uses the equivalence between the $\ell_1$ and the $\ell_2$ norm, \textcolor{black}{and the fact that $|\calM| = 1$ in adversarial team games}.
\end{proof}

\begin{lemma}[Auxiliary]
\label{lem:moreaulinear} 
Let $\vx^* \in \calX$ be an $\epsilon$-approximate stationary point of $\phi_{1/(2\ell)}$, and $\hat{\vx} = \arg \min_{\vx'\in \calX} \max_{\vy \in \calY} U(\vx',\vy)+\ell\norm{\vx^*-\vx'}_2^2$. Further, we define the function $w : \calX \ni \vx \mapsto \max_{\vy\in\calY} \sum_{i=1}^n U(\vx_i,\hat{\vx}_{-i},\vy)$, and we set $w_{1/(2\ell)}(\vx)$ to be the Moreau envelope of $w$. Then,
it holds that 
$$  \textcolor{black}{\norm{ \nabla w_{1/(2\ell)}(\hat{\vx})}_2  \leq 4 \epsilon.}$$
\end{lemma}

\begin{proof}
Since $\norm{\nabla \phi_{1/(2\ell)} (\vx^*)}_2  \leq  \epsilon $,
 we get from~\Cref{lemma:close-prox} that  the proximal point $\hat{\vx}$ of $\vx^*$ satisfies $\norm{\hat{\vx}-\vx^*}_2 \leq \frac{\epsilon}{2\ell}$. Also, there exists a $\vec{\xi} \in\partial \phi (\hat{\vx})$ so that
 $\norm{\vec{\xi}}_2 \leq \epsilon.$
\textcolor{black}{By Danskin's theorem (Fact~\ref{fact:Danskin}),} it follows that $\partial w (\hat{\vx})$ coincides with 
$$\left. \partial \big(\max_y\sum_{i=1}^n U(\vx_i, \hat{\vx}_{-i},\vy)\big) \right|_{\vx=\hat{\vx}} =  \left. \partial \big( \max_y U( {\vx},\vy)\big)\right|_{\vx=\hat{\vx} }.$$ 
As such, there exists a $\vec{\xi} \in\partial w (\hat{\vx})$ so that 
\begin{equation}
\norm{\vec{\xi}}_2\leq \epsilon. \tag{*}
\label{eq:additional}
\end{equation}

We define $\hat{\vx}'\in\calX$ to be the proximal point of $\hat{\vx}$ with respect to $w_{1/(2\ell)}$. Hence, we have
 \begin{align}
     w(\hat{\vx}) &\geq  w_{1/(2\ell)}(\hat{\vx}')  \label{eq:from-moreau-property} \\ 
     &= w(\hat{\vx}') + \ell \| \hat{\vx}' - \hat{\vx} \|_2^2 \label{eq:from-moreau-def} \\
     &\geq \left( w(\hat{\vx}) +  \ell \| \hat{\vx} - \hat{\vx} \|_2^2\right)+ \langle \vec{v}, \hat{\vx}' - \hat{\vx} \rangle + \frac{\ell}{2} \| \hat{\vx}' - \hat{\vx} \|_2^2,  \label{eq:from-strong-convexity}\\
     &= w(\hat{\vx}) + \langle \vec{v}, \hat{\vx}' - \hat{\vx}\rangle + \frac{\ell}{2} \|\hat{\vx}' - \hat{\vx}\|_2^2
     \label{eq:the-conclusion},
 \end{align}
 $~\forall\vec{v} \in \partial \left(w(\vx) + \ell \|\vx - \hat{\vx} \|_2^2 + r(\vx) \right) \Big|_{\vx=\hat{\vx}}.$
 %
\eqref{eq:from-moreau-property} derives from the definition of the Moreau envelope~\citep[Lemma A.1]{lin2020gradient}; \eqref{eq:from-moreau-def} also follows directly from the definition of the Moreau envelope; and \eqref{eq:from-strong-convexity} follows from the fact that $w(\vx)  +\ell \|\vx - \hat{\vx} \|^2+r(\vx)$ \footnote{\textcolor{black}{Recall that $r(\vx) = \infty$ if $\vx\notin\calX$ and zero otherwise. Moreover $\partial r(\vx)$ is $N_{\calX}(\vx)$ if $\vx\in\calX$ and $\vec{0}$ otherwise.}} is $\ell$-strongly convex. Continuing from~\eqref{eq:the-conclusion},
 \begin{align}
     w(\hat{\vx}) \geq w(\hat{\vx}) + \langle \vec{v}, \hat{\vx}'-\hat{\vx}\rangle + \frac{\ell}{2}\| \hat{\vx}' - \hat{\vx}\|_2^2 \iff 
    0 \geq \langle \vec{v}, \hat{\vx}'-\hat{\vx}\rangle + \frac{\ell}{2}\| \hat{\vx}' - \hat{\vx}\|_2^{2}. \label{eq:ineq-for-all-subgrads}
 \end{align}
 Further, it holds that
 $$ \partial \left(w(\vx) +\ell \|\vx - \hat{\vx} \|_2^2 \right) \Bigg|_{\vx=\hat{\vx}} = \partial w(\vx)|_{\vx=\hat{\vx}}.$$
\textcolor{black}{Recall from~\eqref{eq:additional}, there exists $\vec{\xi} \in \partial w(\hat{\vx}): \norm{\vec{\xi}}_2\leq \epsilon$. 
 Thus, replacing $\vec{v}$ with $\vec{\xi}$ in~\eqref{eq:ineq-for-all-subgrads}, and using the latter inequality, we get}
 
    \begin{align*}
        0 &\geq \langle \vec{\xi}, \hat{\vx}' - \hat{\vx} \rangle + \frac{\ell}{2} \| \hat{\vx}' - \hat{\vx} \|_2^2 
        \geq  - \epsilon \| \hat{\vx}' - \hat{\vx} \|_2 + \frac{\ell}{2} \| \hat{\vx}' - \hat{\vx} \|_2^2,
    \end{align*}        
    finally implying that
 $$0 \leq \| \hat{\vx}' - \hat{\vx} \|_2 \leq \frac{ \textcolor{black}{2 \epsilon}}{\ell}.$$
 Concluding, \textcolor{black}{by~\Cref{lemma:close-prox},}
 $$ \textcolor{black}{\norm{  \nabla w_{1/(2\ell)}(\hat{\vx}) }_2 = 2\ell\| \hat{\vx}' - \hat{\vx} \|_2 \leq \textcolor{black}{4 \epsilon}.}$$

\end{proof}

\end{document}